\documentclass[a4paper,UKenglish]{lipics-v2016}

\usepackage{microtype}
\usepackage{amsmath}
\usepackage{amssymb}
\usepackage{booktabs}

\title{Relations Between Greedy and Bit-Optimal LZ77~Encodings}

\author{Dmitry Kosolobov}
\affil{University of Helsinki, Helsinki, Finland\\ \texttt{dkosolobov@mail.ru}}
\authorrunning{D. Kosolobov}
\Copyright{Dmitry Kosolobov}

\EventEditors{Rolf Niedermeier and Brigitte Vall\'ee}
\EventNoEds{2}
\EventLongTitle{35th Symposium on Theoretical Aspects of Computer Science (STACS 2018)}
\EventShortTitle{STACS 2018}
\EventAcronym{STACS}
\EventYear{2018}
\EventDate{February 28 to March 3, 2018}
\EventLocation{Caen, France}
\EventLogo{}
\SeriesVolume{96}
\ArticleNo{47} 

\begin{document}
\maketitle

\begin{abstract}
This paper investigates the size in bits of the LZ77 encoding, which is the most popular and efficient variant of the Lempel--Ziv encodings used in data compression. We prove that, for a wide natural class of variable-length encoders for LZ77 phrases, the size of the greedily constructed LZ77 encoding on constant alphabets is within a factor $O(\frac{\log n}{\log\log\log n})$ of the optimal LZ77 encoding, where $n$ is the length of the processed string. We describe a series of examples showing that, surprisingly, this bound is tight, thus improving both the previously known upper and lower bounds. Further, we obtain a more detailed bound $O(\min\{z, \frac{\log n}{\log\log z}\})$, which uses the number $z$ of phrases in the greedy LZ77 encoding as a parameter, and construct a series of examples showing that this bound is tight even for binary alphabet. We then investigate the problem on non-constant alphabets: we show that the known $O(\log n)$ bound is tight even for alphabets of logarithmic size, and provide tight bounds for some other important cases.
\end{abstract}

\section{Introduction}

The Lempel--Ziv encoding~\cite{LZ77} (LZ77 for short) is one of the most popular and efficient compression techniques used in data compression, stringology, and algorithms in general. The LZ77 encoding lies at the heart of common compressors such as {\tt gzip}, {\tt 7zip}, {\tt pkzip}, {\tt rar}, etc. and serves as a basis for modern compressed text indexes on highly repetitive data (e.g., see~\cite{GGKNP2,KreftNavarroTCS,MakinenNavarro}).

Numerous papers on LZ77 have been published during the last 40 years. In these works, it was proved that LZ77 is superior compared to many other compression schemes both in practice and in theory. For instance, in~\cite{KosarajuManzini,LZ78,WynerZiv} it was shown that LZ77 is asymptotically optimal with respect to different entropy-related measures; further, in~\cite{CharikarEtAl} it was proved that many other reference based encoders (including LZ78~\cite{LZ78}) use polynomially (in the length of the uncompressed data) more space than LZ77 in the worst case and, in a sense, are never significantly better than LZ77. However, many problems related to LZ77 are still not completely solved. In this paper we investigate how good is the popular greedy LZ77 encoder in a class of practically motivated models with variable-length encoders for LZ77 phrases; to formulate the problem that we study more accurately, let us first discuss what is known about different LZ77 encoders.

LZ77 is a dictionary based compression scheme that replaces a string with phrases that are actually references to strings in a dictionary. Each phrase of an LZ77 encoding can be viewed as a triple $\langle d,\ell,c\rangle$, where $\ell$ is the length of the phrase, $d$ is the distance to a string of length $\ell{-}1$ from the dictionary such that this string is a prefix of the phrase, and $c$ is the last letter of the phrase (the precise definition follows); we use the definition from~\cite{LZ77} but all our results can be adapted for the version of LZ77 from~\cite{StorerSzymanski}, in which phrases are encoded by pairs $\langle d,\ell\rangle$ (throughout the paper, we provide the reader with separate remarks in cases where such adaptation is not straightforward). The same string can have many different LZ77 encodings. It is well known that
the greedily constructed LZ77 encoding, which builds the encoding from left to right making each phrase as long as possible during this process, is optimal in the sense that it produces the minimal number of phrases among all LZ77 encodings of this string
(see~\cite{CharikarEtAl,Rytter03,StorerSzymanski}). The same optimality property holds for the versions of LZ77 with ``sliding window''~\cite{CrochemoreLangiuMignosi}, which is a restriction that is important for practical applications.

However, in practice, compressors usually use variable-length encoders for phrases and, in this case, it is not clear whether the greedy LZ77 encoder is optimal in the sense that it outputs the minimal number of bits. The question of finding an optimal LZ77 encoding for variable-length phrase encoders was raised in~\cite{RajpootSahinalp} and the first attempts to solve this problem were given in~\cite{FerraginaNittoVenturini}. The authors of~\cite{FerraginaNittoVenturini} also conducted the first theoretical studies to find how bad is the greedy LZ77 encoding compared to an optimal LZ77 encoding. Such questions make sense only if we state formally which kinds of phrase encoders are used in the LZ77 encoder. As in~\cite{FerraginaNittoVenturini}, we investigate encoders that encode
each phrase $\langle d,\ell,c\rangle$ using $\Theta(\log d + \log \ell + \log c)$ bits\footnote{Throughout the paper all logarithms have base $2$ if it is not explicitly stated otherwise.} (see a more formal discussion below). This class of phrase encoders includes a broad range of practically used encoders and, among others, Elias's~\cite{Elias} and Levenshtein's~\cite{Levenshtein} encoders, which produce asymptotically optimal universal codes for the numbers $d, \ell, c$; we refer the reader to~\cite{FerraginaNittoVenturini} for further discussions on the motivation.

In the described model, there are two ways how to optimize the size of the produced LZ77 encoding. The first way is to minimize $d$ in the triples $\langle d,\ell,c\rangle$. This problem was addressed already in~\cite{FerraginaNittoVenturini} for the greedy LZ77 encoder, where one must find the rightmost occurrence of the referenced part of each phrase; several improvements on this result of~\cite{FerraginaNittoVenturini} and related questions were given in~\cite{AmirLandauUkkonen,BelazzouguiPuglisi,BilleCordingFischerGortz,CrochemoreLangiuMignosi2,Larsson}. The second way is to consider both parameters $\ell$ and $d$, i.e., to build an optimal LZ77 encoding. There are very few works in this direction (see~\cite{CrochemoreEtAl2} and~\cite{FerraginaNittoVenturini}) and there is still a room for improvements in such results. Due to the overall difficulty of the problem of finding an optimal LZ77 encoding, real compressors usually construct an LZ77 encoding greedily. Thus, this raises the following question: how bad can the produced greedy LZ77 encoding be compared to an optimal LZ77 encoding?

For a given string of length $n$, denote by $\mathsf{LZ_{gr}}$ and $\mathsf{LZ_{opt}}$ the sizes in bits of, respectively, the greedily constructed and an optimal LZ77 encodings from the special class of encodings that we consider in this paper (see clarifications in Section~\ref{SectPreliminaries}). We investigate the ratio $\frac{\mathsf{LZ_{gr}}}{\mathsf{LZ_{opt}}}$. Upper bounds on this ratio are provided in terms of the parameters $n$, $z$, and $\sigma$, where $z$ is the number of phrases in the greedy LZ77 encoding of the considered string (it is well known that any other LZ77 encoding contains at least $z$ phrases; see~\cite{CharikarEtAl,Rytter03,StorerSzymanski}) and $\sigma$ is the alphabet size. We are also interested in upper bounds that use only the parameter $n$. In~\cite{FerraginaNittoVenturini} it was proved that $\frac{\mathsf{LZ_{gr}}}{\mathsf{LZ_{opt}}} = O(\log n)$ and there is a series of examples on which $\frac{\mathsf{LZ_{gr}}}{\mathsf{LZ_{opt}}} = \Omega(\frac{\log n}{\log\log n})$. In this paper we improve these results and our bounds in many cases are tight in the sense that there are series of examples on which these bound are attained; our main contributions are summarized in Table~\ref{tbl:results}.

\begin{table}[ht]
\caption{Upper bounds on $\mathsf{LZ_{gr}}/ \mathsf{LZ_{opt}}$; tight bounds are denoted by $\Theta$.}
\begin{center}
\begin{tabular}{r|c|c}\hline
~      & parameter $n$ & parameters $n, z, \sigma$ \\\midrule
$\sigma = O(1)$ & $\Theta(\frac{\log n}{\log\log\log n})$ & $\Theta(\min\{z, \frac{\log n}{\log\log z}\})$   \\
arbitrary $\sigma$    & $\Theta(\log n)$                        & $O(\min\{z, \frac{\log n}{\log\log_\sigma z}\})$ \\\hline      
\end{tabular}
\label{tbl:results}
\end{center}
\end{table}

First, we study the case of constant alphabets and completely solve it. Namely, in Theorem~\ref{MainTheorem}, we find the following detailed upper bound on the ratio $\frac{\mathsf{LZ_{gr}}}{\mathsf{LZ_{opt}}}$ (note that this bound is also applicable for arbitrary alphabets): $\frac{\mathsf{LZ_{gr}}}{\mathsf{LZ_{opt}}} = O(\min\{z, \frac{\log n}{\log\log_\sigma z}\}).$
In the case of constant alphabets this upper bound degenerates to $O(\min\{z, \frac{\log n}{\log\log z}\})$. In Theorem~\ref{ExampleTheorem} we construct a series of examples on the binary alphabet showing that this simplified bound is tight, thus closing the problem for constant alphabets. Theorem~\ref{ExampleTheorem} actually provides a more elaborate lower bound $\Omega(\min\{z, \frac{\log n}{\log\log_\sigma z + \log\sigma}\})$, which is applicable for arbitrary alphabets. From these general results, we deduce in Corollary~\ref{OnlyNbound} that $\frac{\mathsf{LZ_{gr}}}{\mathsf{LZ_{opt}}} = O(\frac{\log n}{\log\log\log n})$ for constant alphabets, and this upper bound is tight.

Then, we consider the case of arbitrary alphabets. It is shown in Theorem~\ref{ExampleTheorem2} that the upper bound $O(\log n)$ on the ratio $\frac{\mathsf{LZ_{gr}}}{\mathsf{LZ_{opt}}}$ is tight even if the input alphabet has logarithmic size. Thus, we solve the problem in the general case and find that the tight upper bounds, expressed in terms of $n$, for constant and arbitrary alphabets differ by $\Theta(\log\log\log n)$ factor.

As a side note, for polylogarithmic alphabets and $z \ge 2^{\log^\epsilon n}$, where $\epsilon > 0$ is an arbitrary constant, we obtain in Corollary~\ref{MainIsTight2} the upper bound $O(\frac{\log n}{\log\log n})$ and show that this bound is tight for such alphabets and such $z$. Informally, the strings for which the condition $z \ge 2^{\log^\epsilon n}$ holds (which includes the case $z \ge n^\delta$, where $\delta > 0$ is an arbitrary constant) can be called ``non-extremely compressible'' strings. Thus, we, in a sense, solve the problem in the arguably most important case of ``non-extremely compressible'' strings drawn from polylogarithmic alphabets.

The paper is organized as follows. In the following Section~\ref{SectPreliminaries} we introduce some basic notions used throughout the text and, in particular, formally define LZ77 parsings and encodings. Section~\ref{SectUpperBound} describes a detailed upper bound on the ratio of the sizes in bits of the greedy and optimal LZ77 encodings. In Section~\ref{SectLowerBound} it is shown that, on constant alphabets, this bound is tight. The material of these two sections provides a complete solution of the problem for constant alphabets, which turns out to be quite simple. We then consider arbitrary alphabets in Section~\ref{SectArbitraryAlphabet} and find tight bounds for several important cases, including the general case of arbitrary alphabet and arbitrary $z$, for which, as it turns out, the known $O(\log n)$ bound is tight. Finally, we conclude with some remarks and open problems in Section~\ref{SectConclusion}.

\section{Preliminaries}\label{SectPreliminaries}

A \emph{string $s$} over an alphabet $\Sigma$ is a map $\{1,2,\ldots,n\} \to \Sigma$, where $n$ is referred to as the \emph{length of $s$}, denoted by $|s|$. In this paper we assume that the alphabet is a set of non-negative integers that are less than or equal to $n$, which is a common and natural assumption in the problem under investigation. We write $s[i]$ for the $i$th letter of $s$ and $s[i..j]$ for $s[i]s[i{+}1]\cdots s[j]$. A string $u$ is a \emph{substring} of $s$ if $u = s[i..j]$ for some $i$ and $j$; the pair $(i,j)$ is not necessarily unique and we say that $i$ specifies an \emph{occurrence} of $u$ in $s$ starting at position $i$. A substring $s[1..j]$ (resp., $s[i..n]$) is a \emph{prefix} (resp. \emph{suffix}) of $s$. We say that substrings $s[i..j]$ and $s[i'..j']$ \emph{overlap} if $j \ge i'$ and $i \le j'$. For any $i,j$, the set $\{k\in \mathbb{Z} \colon i \le k \le j\}$ (possibly empty) is denoted by $[i..j]$.


An \emph{LZ77 parsing} of a given string $s$ is a parsing $s = f_1f_2\cdots f_z$ such that all the strings $f_1, \ldots, f_z$ (called \emph{phrases}) are non-empty and, for any $i \in [1..z]$, either $f_i$ is a letter, or $|f_i| > 1$ and the string $f_i[1..|f_i|{-}1]$ has an earlier occurrence starting at some position $j \le |f_1f_2\cdots f_{i-1}|$ (note that this occurrence can overlap $f_i$).

The \emph{greedy LZ77 parsing} is a special LZ77 parsing built by the greedy procedure that constructs all phrases from left to right by choosing each phrase $f_i$ as the longest substring starting at given position such that $f_i[1..|f_i|{-}1]$ has an earlier occurrence in the string (see~\cite{LZ77}). For instance, the greedy LZ77 parsing of the string $s = abababbbaba$ is $a.b.ababb.baba$. The following lemma is straightforward.

\begin{lemma}
All phrases in the greedy LZ77 parsing of a given string (except, possibly, for the last phrase) are distinct.
\label{DistinctPhrases}
\end{lemma}

It is also well-known that, for a given string, the greedy LZ77 parsing has the minimal number of phrases among all LZ77 parsings (e.g., see~\cite{CharikarEtAl,Rytter03,StorerSzymanski}). This implies that, when each phrase of the parsing is encoded by a fixed number of bits, the greedy LZ77 parsing is optimal, i.e., it produces an encoding of the minimal size in bit. However, the greedy LZ77 parsing does not necessarily produce an encoding of the minimal size when one uses a variable-length encoder for phrases; the latter is usually the case in most common compressors. Let us clarify what kinds of variable-length phrase encoders we are to consider in this paper.

A given LZ77 parsing $f_1f_2\cdots f_z$ is encoded as follows. Each phrase $f_i$ is represented by a triple $\langle d, \ell, c\rangle$, where $\ell = |f_i|$, $c = f_i[|f_i|]$, and $d = |f_1f_2\cdots f_{i-1}| - j$ for $j$ that is the position of an earlier occurrence of $f_i[1..|f_i|{-}1]$ (assuming that $d = 0$ if $|f_i| = 1$). We choose three encoders $e_{d}, e_{\ell}, e_c$, each of which maps non-negative integers to bit strings. We then transform each triple $\langle d,\ell,c\rangle$ into the binary string $e_{d}(d)e_{\ell}(\ell)e_{c}(c)$ and concatenate all these binary strings, thus producing an \emph{LZ77 encoding} corresponding to the given LZ77 parsing.

In this paper we consider only encoders $e_{d}, e_{\ell}, e_c$ that map any positive integer $x$ to a bit string of length $\Theta(\log(x+1))$. This family of encoders includes most widely used encoders such as Elias's~\cite{Elias} and Levenshtein's~\cite{Levenshtein} ones (see~\cite{FerraginaNittoVenturini} for further motivation). We fix three encoders $e_{d}, e_{\ell}, e_c$ satisfying the above property and, hereafter, assume that all considered LZ77 encodings are obtained using these $e_{d}, e_{\ell}, e_c$.

We say that an LZ77 encoding is \emph{optimal} if it has the minimal size in bits. It is shown below that, unlike the case of fixed-length phrase encoders, for the family of phrase encoders under investigation, the LZ77 encoding generated by the greedy LZ77 parsing (which is called the \emph{greedy LZ77 encoding}) is not necessarily optimal. Among all possible greedy LZ77 encodings we always consider those that occupy the minimal number of bits; usually, such encoding is obtained by the minimization of the numbers $d$ in the triples $\langle d,\ell,c\rangle$ representing the phrases of the greedy LZ77 parsing.

\begin{remark}
Most common compressors actually use a different variant of the LZ77 parsing (which was introduced in~\cite{StorerSzymanski}), defining each phrase $f_i$ as either a letter or a string that has an earlier occurrence (note that in the definition of LZ77 parsings only the prefix $f_i[1..|f_i|{-}1]$ of $f_i$ must have an earlier occurrence). We call this variant a \emph{nonclassical LZ77 parsing} (as it differs from the original parsing proposed in~\cite{LZ77}). The \emph{greedy nonclassical LZ77 parsing} is defined by analogy with the greedy LZ77 parsing. In encoding corresponding to a nonclassical LZ77 parsing each phrase is represented either by a pair $\langle d,\ell\rangle$ that is defined analogously to the triples $\langle d,\ell,c\rangle$, or by one letter. This variant of LZ77 is very similar to the one that we investigate and, moreover, all our results can be adapted for this variant. In the sequel, we provide separate remarks that explicitly show how to generalize our results to nonclassical LZ77 parsings if it is not straightforward.
\end{remark}

\section{Upper Bound}\label{SectUpperBound}

Our proof of the upper bound on the ratio between the sizes of the greedy and optimal LZ77 encodings is as follows: first, we obtain an upper bound $U$ on the size of the greedy LZ77 encoding, then we find a lower bound $L$ on the size of any LZ77 encoding, and finally, we derive the estimation $\frac{U}{L}$ on the ratio. The details follow.

Let $s$ be a string of length $n$. Recall that any letter of $s$ is an integer from the range $[0..n]$. Based on the above mentioned properties of the phrase encoders $e_{d}, e_{\ell}, e_c$, one can easily show that each phrase of any LZ77 encoding of $s$ occupies $O(\log n)$ bits. Therefore, we obtain the following upper bound on the size of the greedy LZ77 encoding.

\begin{lemma}
Let $\mathsf{LZ_{gr}}$ be the size in bits of the greedy LZ77 encoding of a given string of length $n$. Then, we have $\mathsf{LZ_{gr}} = O(z\log n)$, where $z$ is the number of phrases in the encoding.\label{GreedyLZ77upper}
\end{lemma}

The lower bound on any LZ77 encoding is more complicated. Lemmas~\ref{TechLemma},~\ref{LZ77intersect},~\ref{LongPhrasesSet} below are well known but we, nevertheless, provide their proofs for the sake of completeness.

\begin{lemma}
For any positive integers $t, t_1, \ldots, t_k$ such that $\sum_{i=1}^k t_i \ge t$, we have $\sum_{i=1}^k\log t_i \ge \log(t-k+1)$.
\label{TechLemma}
\end{lemma}
\begin{proof}
Note that $\sum_{i=1}^k\log t_i = \log\prod_{i=1}^k t_i$. Since for any $t_j$ and $t_{j'}$ such that $t_j \ge t_{j'}$, we have $(t_j + 1)(t_{j'} - 1) = t_jt_{j'} - (t_j - t_{j'} + 1) < t_jt_{j'}$, the product $\prod_{i=1}^k t_i$ is minimized when $t_1 = t - k + 1$ and $t_2 = t_3 = \cdots = t_{k} = 1$ (recall that every number $t_i$ must be a positive integer). Therefore, we obtain $\sum_{i=1}^k\log t_i \ge \log(t-k+1)$.
\end{proof}


\begin{lemma}
Any phrase of an LZ77 parsing of a string can overlap with at most two phrases of the greedy LZ77 parsing of the same string.
\label{LZ77intersect}
\end{lemma}
\begin{proof}
Suppose, for the sake of contradiction, that a phrase $f$ of an LZ77 parsing overlaps with at least three phrases of the greedy LZ77 parsing. Then, $f[1..|f|{-}1]$ must contain a phrase $f'$ of the greedy LZ77 parsing as a proper substring. But then the string $f'$ occurs in an earlier occurrence of the string $f[1..|f|{-}1]$ and, therefore, the greedy construction procedure could choose a longer phrase during the construction of the phrase $f'$, which is a contradiction.
\end{proof}

\begin{lemma}
In the greedy LZ77 parsing of any string of length $n$ over an alphabet of size $\sigma \ge 2$, at least $z - 2\sqrt{z}$ phrases have length ${\ge}\frac{1}{2}\log_\sigma z$, where $z$ is the number of phrases.\label{LongPhrasesSet}
\end{lemma}
\begin{proof}
Denote by $f_1f_2\cdots f_z$ the greedy LZ77 parsing of a given string of length $n$ over an alphabet of size $\sigma$. By Lemma~\ref{DistinctPhrases}, all the phrases $f_1, \ldots, f_{z-1}$ are distinct. Therefore, for any $\ell > 0$, at most $\sum_{i=0}^{\ell} \sigma^i = \frac{\sigma^{\ell+1} - 1}{\sigma - 1}$ of these phrases have length at most $\ell$. Since for any $\ell < \frac{1}2\log_\sigma z$, we have $\sum_{i=0}^{\ell}\sigma^i < \frac{\sqrt{z}\sigma - 1}{\sigma - 1}$, the number of phrases with length at least $\frac{1}2 \log_\sigma z$ must be greater than $(z - 1) - \frac{\sqrt{z}\sigma - 1}{\sigma - 1}$. Thus, it remains to prove that $1 + \frac{\sqrt{z}\sigma - 1}{\sigma - 1} \le 2\sqrt{z}$. It is easy to show that, for $\sigma \ge 2$, the function $\frac{\sqrt{z}\sigma - 1}{\sigma - 1}$ decreases as $\sigma$ grows. Hence, we deduce $1 + \frac{\sqrt{z}\sigma - 1}{\sigma - 1} \le 1 + \frac{2\sqrt{z} - 1}{2 - 1} = 2\sqrt{z}$.
\end{proof}

\begin{lemma}
Let $\mathsf{LZ_{opt}}$ be the size in bits of an optimal LZ77 encoding of a string of length $n$ over an alphabet of size $\sigma \ge 2$. Then, we have $\mathsf{LZ_{opt}} = \Omega(\log n + z\log\log_\sigma z)$, where $z$ is the number of phrases in the greedy LZ77 parsing of this string.
\label{OptLZ77lower}
\end{lemma}
\begin{proof}
Denote by $f_1f_2\cdots f_{z'}$ the LZ77 parsing corresponding to an optimal LZ77 encoding of the string under consideration.
By the definition of the phrase encoders, we have $\mathsf{LZ_{opt}} \ge \Omega(\sum_{i=1}^{z'} \log |f_i|)$.
It follows from Lemma~\ref{TechLemma} that $\mathsf{LZ_{opt}} \ge \Omega(\log(n - z'))$. Since, obviously, $\mathsf{LZ_{opt}} \ge z'$, the latter implies $\mathsf{LZ_{opt}} \ge \Omega(z' + \log(n - z')) \ge \Omega(\log n)$.

Denote by $f'_1f'_2\cdots f'_z$ the greedy LZ77 parsing of the same string. Let $S$ be the set of all phrases in this parsing with lengths at least $\frac{1}2 \log_\sigma z$. By Lemma~\ref{LongPhrasesSet}, we have $|S| \ge z - 2\sqrt{z} = \Theta(z)$. Consider a phrase $f' \in S$. Let $f_g, f_{g+1}, \ldots, f_{h}$ be all phrases in the parsing $f_1f_2\cdots f_{z'}$ that overlap with the phrase $f'$. Since $|f_gf_{g+1}\cdots f_h| \ge |f'|$, Lemma~\ref{TechLemma} implies that $(h - g) + \log|f_g| + \log|f_{g+1}| + \cdots + \log|f_h| \ge (h - g) + \log(|f'| - (h - g)) \ge \Omega(\log|f'|)$. Thus, the encodings of the phrases $f_g, f_{g+1}, \ldots, f_h$ all together occupy $\Omega(\log|f'|)$ bits. By Lemma~\ref{LZ77intersect}, any phrase $f_i$ of the parsing $f_1\cdots f_{z'}$ overlaps with at most two phrases of the parsing $f'_1\cdots f'_z$. Therefore, the encodings of all phrases $f_1, \ldots, f_{z'}$ occupy $\frac{1}{2}\Omega(\sum_{f' \in S} \log|f'|) \ge \Omega(|S|\log\log_\sigma z) = \Omega(z\log\log_\sigma z)$ overall bits.
\end{proof}

\begin{theorem}
Let $z$ be the number of phrases in the greedy LZ77 parsing of a given string of length $n$ drawn from an alphabet of size $\sigma$. Denote by $\mathsf{LZ_{gr}}$ and $\mathsf{LZ_{opt}}$ the sizes in bits of, respectively, the greedy and optimal LZ77 encodings of this string. Then, we have $\frac{\mathsf{LZ_{gr}}}{\mathsf{LZ_{opt}}} = O(\min\{z, \frac{\log n}{\log\log_\sigma z}\})$.\label{MainTheorem}
\end{theorem}
\begin{proof}
By Lemmas~\ref{GreedyLZ77upper} and~\ref{OptLZ77lower}, $\frac{\mathsf{LZ_{gr}}}{\mathsf{LZ_{opt}}} \le \frac{O(z\log n)}{\Omega(\log n + z\log\log_\sigma z)} = O(\frac{z\log n}{\log n + z\log\log_\sigma z})$. Since $\frac{z\log n}{\log n + z\log\log_\sigma z} \le \frac{z\log n}{\log n} = z$ and $\frac{z\log n}{\log n + z\log\log_\sigma z} \le \frac{\log n}{\log\log_\sigma z}$, the result follows.
\end{proof}

\begin{corollary}
For constant alphabet, $\frac{\mathsf{LZ_{gr}}}{\mathsf{LZ_{opt}}} = O(\frac{\log n}{\log\log\log n})$.\label{OnlyNbound}
\end{corollary}
\begin{proof}
We have $\frac{\mathsf{LZ_{gr}}}{\mathsf{LZ_{opt}}} = O(\min\{z, \frac{\log n}{\log\log z}\})$ due to Theorem~\ref{MainTheorem}. The functions $z \mapsto z$ and $z \mapsto \frac{\log n}{\log\log z}$, respectively, increase and decrease as $z$ grows. Therefore, the maximum of the function $\min\{z, \frac{\log n}{\log\log z}\}$ is reached when $z = \frac{\log n}{\log\log z}$. Solving this equation, we obtain $z = \Theta(\frac{\log n}{\log\log\log n})$, which proves the result.
\end{proof}

\begin{remark}
To generalize the described results to nonclassical LZ77 parsings, one should use, instead of Lemma~\ref{DistinctPhrases}, the following straightforward lemma.
\begin{lemma}
Suppose that $s = f_1f_2\cdots f_z$ is the greedy nonclassical LZ77 parsing of a given string $s$; then, all the strings $f_i\cdot f_{i+1}[1]$, for $i \in [1..z{-}1]$, are distinct.\label{DistinctPhrases2}
\end{lemma}

The rest can be easily reconstructed by analogy.
\end{remark}

\section{Lower Bound}\label{SectLowerBound}

We now construct a series of example showing that, for several important cases, the upper bound given in Theorem~\ref{MainTheorem} is tight. In particular, on constant alphabets, i.e., when $\sigma = O(1)$, Theorem~\ref{ExampleTheorem} complements Theorem~\ref{MainTheorem} showing that the bound $O(\min\{z, \frac{\log n}{\log\log z}\})$ is tight. Further, putting $z = \frac{\log n}{\log\log\log n}$ and $\sigma = 2$ in Theorem~\ref{ExampleTheorem}, we show that the upper bound given in Corollary~\ref{OnlyNbound} is tight.

\begin{theorem}
For any given integers $n > 1$, $\sigma \in [2..n]$, and $z \in [\sigma..\frac{n}{\log_\sigma n}]$, there is a string of length $n$ over an alphabet of size $\sigma$ such that the number of phrases in the greedy LZ77 parsing of this string is $\Theta(z)$ and the sizes $\mathsf{LZ_{gr}}$ and $\mathsf{LZ_{opt}}$ of, respectively, the greedy and optimal LZ77 encodings of this string are related as $\frac{\mathsf{LZ_{gr}}}{\mathsf{LZ_{opt}}} \ge \Omega(\min\{z, \frac{\log n}{\log\log_\sigma z + \log\sigma}\})$.
\label{ExampleTheorem}
\end{theorem}
\begin{proof}
If $\sigma \ge n/4$, then any LZ77 encoding of a string of length $n$ containing $\sigma$ distinct letters obviously occupies $\Theta(\sigma\log\sigma) = \Theta(n\log n)$ bits and, hence, the statement of the theorem, which degenerates to $\frac{\mathsf{LZ_{gr}}}{\mathsf{LZ_{opt}}} \ge \Omega(1)$, trivially holds. Assume that $\sigma < n/4$.


We first consider the case $\sigma \ge 3$ as it is simpler. Suppose that the alphabet is the set $[1..\sigma]$. Denote $b = 1$ and $\tau = \sigma - 1$ ($b$ is a special letter-separator with small code and $\tau$ is the size of the set $[1..\sigma] \setminus \{b\} = [2..\sigma]$). Let $m$ be the minimal integer such that $\tau^m \ge z$, i.e., $m = \lceil\log_{\tau} z\rceil$. Note that $m = \Theta(\log_\sigma z)$. In~\cite{Cohn} it is shown that all $\tau^m$ possible strings of length $m$ over the alphabet $[2..\sigma]$ can be arranged in a sequence $s_1, s_2, \ldots, s_{\tau^m}$ (called a \emph{$\tau$-ary Gray code}~\cite{Cohn,Gray}) such that, for any $i \in [2..\tau^m]$, the strings $s_{i-1}$ and $s_i$ differ in exactly one position. Moreover, we can choose such sequence so that $s_{\tau^m} = a^m$, where $a$ is an arbitrary letter from $[2..\sigma]$.

Let $k$ and $\ell$ be positive integers such that $k < \tau^m$ and $\ell > m$. Our example is the following string (the numbers $k$ and $\ell$ will be adjusted below so that $k = \Theta(z)$ and $\ell \ge \frac{1}{2}n$):
$$
s = s_1 s_2 \cdots s_k\cdot a^{\ell}\cdot bs_1bs_2b \cdots s_k b.
$$

Let us consider the greedy LZ77 parsing of $s$ and the corresponding greedy LZ77 encoding. Since the letter $b$ first occurs in the substring $a^\ell b$, the greedy construction procedure builds the parsing of $s_1 b s_2 b \cdots s_k b$ starting from the first position of this substring. Since $k < \tau^m$ and $s_{\tau^m} = a^m$, it follows from the definition of the sequence $s_1, \ldots, s_{k}$ that, for any $i\in[1..k]$, the longest prefix of the string $s_i b s_{i+1} b \cdots s_k b$ that has an earlier occurrence in $s$ is $s_i$ and this earlier occurrence is a substring of the prefix $s_1 s_2 \cdots s_ka^m$ of $s$. Therefore, the greedy algorithm decomposes the suffix $s_1 b s_2 b \cdots s_k b$ into $k$ phrases $s_i b$, for $i \in [1..k]$. It is easy to see that each of these phrases is encoded in $\Omega(\log\ell)$ bits (this is the number of bits required to encode the distance between the phrase and its earlier occurrence). Hence, the size in bits of the greedy LZ77 encoding of $s$ is $\mathsf{LZ_{gr}} \ge \Omega(k\log\ell)$.

Now let us consider a better encoding of the same string $s$. For simplicity, we omit the description of the encoding of the prefix $s_1 s_2 \cdots s_k$ as it is very similar to the encoding of the suffix $s_1 b s_2 b \cdots s_k b$ discussed below. First, we parse the substring $a^{\ell}b$ into two phrases $a$ and $a^{\ell-1}b$, which are encoded in $O(\log\ell + \log\sigma)$ bits (the referenced part $a^{\ell-1}$ of $a^{\ell-1}b$ is self-referential). Then, we encode the substring $s_1 b$ as in the greedy approach by one phrase taking $O(\log\ell)$ bits (recall that $\ell > m$ and $b = 1$ and, hence, the length $|s_1b| = m + 1$ and the letter $b$ are encoded in $O(\log\ell)$ bits). Now we consecutively encode each substring $s_i b$, for $i \in [2..k]$, as follows. Suppose that the strings $s_i$ and $s_{i-1}$ differ at position $j$, i.e., $s_{i-1}[1..j{-}1] = s_i[1..j{-}1]$ and $s_{i-1}[j{+}1 .. m] = s_i[j{+}1 .. m]$. We decompose $s_i b$ into two phrases $s_i[1..j]$ and $s_i[j{+}1..m]b$. Since the strings $s_i[1..j{-}1]$ and $s_i[j{+}1..m]$ both are substrings of the string $s_{i-1}$ and have length $O(m)$, the encoding of the produced two phrases occupies $O(\log m + \log\sigma)$ bits. Hence, the whole suffix $s_1bs_2b\cdots s_kb$ can be encoded in $O(k\log m + k\log\sigma)$ bits; the prefix $s_1 s_2 \cdots s_k$ can be encoded similarly in $O(k\log m + k\log\sigma)$ bits. Thus, we obtain an encoding of the string $s$ that occupies $O(\log\ell + k\log m + k\log\sigma)$ bits. Therefore, the size in bits of the optimal LZ77 encoding of $s$ is $\mathsf{LZ_{opt}} = O(\log\ell + k\log m + k\log\sigma)$.

Recall that $m = \Theta(\log_\sigma z)$. Combining the estimations on $\mathsf{LZ_{gr}}$ and $\mathsf{LZ_{opt}}$, we obtain $\frac{\mathsf{LZ_{gr}}}{\mathsf{LZ_{opt}}} \ge \frac{\Omega(k\log\ell)}{O(\log\ell + k(\log m + \log\sigma))} \ge \Omega(\frac{k\log\ell}{\log\ell + k(\log\log_\sigma z + \log\sigma)})$. Since $\frac{k\log\ell}{\log\ell + k(\log\log_\sigma z + \log\sigma)} \ge \frac{k\log\ell}{2\cdot\max\{\log\ell, k(\log\log_\sigma z + \log\sigma)\}} = \frac{1}{2} \min\{k, \frac{\log\ell}{\log\log_\sigma z + \log\sigma}\}$, we obtain $\frac{\mathsf{LZ_{gr}}}{\mathsf{LZ_{opt}}} \ge \Omega(\min\{k, \frac{\log\ell}{\log\log_\sigma z + \log\sigma}\})$. Note that the number of phrases in the greedy LZ77 parsing of $s$ is $\Theta(k)$ and $|s| = \ell + 1 + k(2m + 1)$. We put $\ell = n - k(2m + 1) - 1$ so that $|s| = n$. Since $z \in [2..\frac{n}{\log_\sigma n}]$ and $m = \Theta(\log_{\sigma} z)$, we have $k(2m + 1) \le O(n)$ if $k = \Theta(z)$. Then, it is straightforward that the parameter $k$ can be chosen so that $k = \Theta(z)$ and $\ell = n - k(2m + 1) - 1 \ge \frac{1}{2}n$. Hence, we derive $\frac{\mathsf{LZ_{gr}}}{\mathsf{LZ_{opt}}} \ge \Omega(\min\{z, \frac{\log n}{\log\log_\sigma z + \log\sigma}\})$. (If not all letters of the alphabet $[1..\sigma]$ indeed occur in the constructed string, we append all unused letters to the end of $s$ and reduce $\ell$ appropriately; as $\sigma < n/4$, we have $\ell \ge \frac{1}{4}n$ in the end.)


Now assume that $\sigma = 2$.
Let $\{0,1\}$ be the alphabet.
Similarly to the above analysis, we fix a sequence $s_1, \ldots, s_{2^m}$ of all binary strings of length $m = \lceil\log z\rceil$ such that, for $i \in [2..2^m]$, $s_{i-1}$ and $s_i$ differ in exactly one position, and we choose two parameters $\ell > 4m$ and $k < 2^m$, which will be adjusted later so that $\ell \ge \frac{1}2n$ and $k = \Theta(z)$. It is well known that one can fix the sequence $s_1, \ldots, s_{2^m}$ so that $s_{2^m} = 0^m$. Our example is defined as follows:
$$
s = s_1 0^m 1 s_2 0^m 1 \cdots s_k 0^m 1 0^{\ell}1 s_1 0^m 1 c_1 s_2 0^m 1 c_2\cdots s_k 0^m 1 c_k,
$$
where $c_k = 1$ and, for $i \in [1..k{-}1]$, $c_i = 0$ if $s_{i+1}[1] = 1$, and $c_i = 1$ otherwise.

Since, for any $i \in [1..k]$, $s_i \ne 0^m$ (as $s_i = 0^m$ iff $i = 2^m$, and $k < 2^m$) and $\ell > 4m$, the greedy LZ77 parser necessarily makes a phrase that is a suffix of the substring $0^\ell1$ and, then, parses the suffix $s_1 0^m 1 c_1 s_2 0^m 1 c_2\cdots s_k 0^m 1 c_k$ from the first position. It is straightforward that, for any $i \in [1..k]$, the string $0^m1$ has only one occurrence in the strings $1s_i0^m1$ and $1c_{i-1}s_i0^m1$ (for $i > 1$). Therefore, for any $i \in [1..k]$, the string $s_i 0^m 1$ has only one occurrence in the prefix $s_1 0^m 1 s_2 0^m 1 \cdots s_k 0^m 1$ and the string $s_i 0^m 1 c_i$ has only one occurrence in the whole string $s$. Then, the greedy parser parses the suffix $s_1 0^m 1 c_1 s_2 0^m 1 c_2\cdots s_k 0^m 1 c_k$ into $k$ phrases $s_i 0^m 1 c_i$, for $i\in[1..k]$. This parsing produces an encoding of size $\Omega(k\log\ell)$ bits. At the same time, there is an LZ77 encoding for $s$ of size $O(\log\ell + k\log m)$ bits. The further analysis is very similar to the analysis of the case $\sigma \ge 3$: we put $\ell = n - k(4m + 3) - 1$ so that $|s| = n$, and we adjust $k$ so that $k = \Theta(z)$ and $\ell \ge \frac{1}{2}n$, which is possible because $m \le \log z + 1$ and $z \le \frac{n}{\log n}$. We omit the details as they are analogous.
\end{proof}

\begin{remark}
The condition $\sigma \le z \le \frac{n}{\log_\sigma n}$ from Theorem~\ref{ExampleTheorem} is justified by the following observations. First, it is obvious that any LZ77 parsing has at least $\sigma$ phrases and, hence, the inequality $\sigma \le z$ holds. Secondly, by Lemma~\ref{LongPhrasesSet}, at least $z - 2\sqrt{z}$ phrases in the greedy LZ77 parsing have length at least $\frac{1}{2}\log_\sigma z$, where $z$ is the total number of phrases; hence, we obtain $z\log_\sigma z \le O(n)$ and, solving this inequality, $z = O(\frac{n}{\log_\sigma n})$, which justifies the condition $z \le \frac{n}{\log_\sigma n}$.
\end{remark}

\begin{remark}
Let us sketch the way in which the constructions from the proof of Theorem~\ref{ExampleTheorem} can be adapted to nonclassical LZ77 encodings. For the case $\sigma \ge 3$, the corresponding string is as follows (the notation is from the proof of Theorem~\ref{ExampleTheorem}):
$$
s = bs_1bs_2\cdots bs_kb\cdot a^\ell\cdot bs_1bbs_2bb\cdots bbs_kb.
$$
The suffix $bs_1bbs_2bb\cdots bbs_kb$ of this string is greedily parsed into the phrases $bs_ib$, for $i \in [1..k]$. For the case $\sigma = 2$, the corresponding string is as follows:
$$
s = 10s_1\alpha 10s_2\alpha 1\cdots 10s_k\alpha\cdot 0^\ell\cdot 1 0s_1\alpha 0s_2\alpha 0\cdots 0s_k\alpha 0,
$$
where $\alpha = 0^{m+1}1$. The suffix $1 0s_1\alpha 0s_2\alpha 0\cdots 0s_k\alpha 0$ of $s$ is greedily parsed into the phrases $10s_1\alpha$ and $0s_i\alpha$, for $i \in [2..k]$. We omit the detailed analysis as it is analogous to the analysis in the proof of Theorem~\ref{ExampleTheorem}.
\end{remark}

\section{Arbitrary Alphabets}\label{SectArbitraryAlphabet}

The following corollary shows that, in the case of ``non-extremely compressible'' string ($z \ge 2^{\log^\epsilon n}$) over a polylogarithmic alphabet ($\sigma \le \log^{O(1)} n$), which is arguably the most important case for practice, the upper and lower bounds from Theorems~\ref{MainTheorem} and~\ref{ExampleTheorem} degenerate to $\Theta(\frac{\log n}{\log\log n})$ and, hence, are tight. (Note that $2^{\log^\epsilon n} = o(n^\delta)$ for any fixed constants $\epsilon \in (0,1)$ and $\delta \in (0,1)$.)

\begin{corollary}
Let $z$ be the number of phrases in the greedy LZ77 parsing of a given string of length $n$ drawn from an alphabet of size $\sigma$.
Suppose that $\sigma \le \log^{O(1)} n$ and $z \ge 2^{\log^\epsilon n}$, for a fixed constant $\epsilon \in (0,1)$. Denote by $\mathsf{LZ_{gr}}$ and $\mathsf{LZ_{opt}}$ the sizes in bits of, respectively, the greedy and optimal LZ77 encodings of this string. Then, we have $\frac{\mathsf{LZ_{gr}}}{\mathsf{LZ_{opt}}} \le O(\frac{\log n}{\log\log n})$ and this upper bound is tight.
\label{MainIsTight2}
\end{corollary}
\begin{proof}
The result follows from Theorems~\ref{MainTheorem} and~\ref{ExampleTheorem} since $\log\log n \ge \log\log_\sigma z \ge \log\frac{\log^\epsilon n}{O(\log\log n)} = \Theta(\log\log n)$.
\end{proof}

Now let us consider bounds on the ratio $\frac{\mathsf{LZ_{gr}}}{\mathsf{LZ_{opt}}}$ that are independent of the parameters $z$ and $\sigma$.

In~\cite{FerraginaNittoVenturini} it was proved that $O(\log n)$ is an upper bound on the ratio $\frac{\mathsf{LZ_{gr}}}{\mathsf{LZ_{opt}}}$. It turns out that this bound is tight on sufficiently large non-constant alphabets. Precisely, a series of examples on which $\frac{\mathsf{LZ_{gr}}}{\mathsf{LZ_{opt}}} = \Omega(\log n)$ can be constructed on an alphabet of size $O(\log n)$. Therefore, the upper bound $O(\frac{\log n}{\log\log\log n})$ on the ratio $\frac{\mathsf{LZ_{gr}}}{\mathsf{LZ_{opt}}}$, which, by Corollary~\ref{OnlyNbound}, holds for constant alphabets and is tight, does not hold, in general, even for alphabets of logarithmic size. In examples showing this, we use the following well-known combinatorial structure.

A \emph{Steiner system} $S(t,k,n)$ is a set $S$ of size $n$ and a family of $k$-element subsets of $S$, called \emph{blocks}, such that each subset of $S$ of size $t$ is contained in exactly one block. We are particularly interested in the Steiner systems $S(2, 2^{2^{i-1}}, 2^{2^i})$, which can be constructed for any positive integers $i$ (the structure is realized on a finite affine plane of order $2^{2^{i-1}}$ and the blocks are lines in the plane; see~\cite{ColbournDinitz}).
It is well known that the number of blocks in the Steiner system $S(2, 2^{2^{i-1}}, 2^{2^i})$ is ${2^{2^i} \choose 2} / {2^{2^{i-1}} \choose 2}$.

\begin{theorem}
For any integer $n > 1$, there is a string of length $n$ over an alphabet of size $O(\log n)$ such that the sizes $\mathsf{LZ_{gr}}$ and $\mathsf{LZ_{opt}}$ of, respectively, the greedy and optimal LZ77 encodings of this string are related as $\frac{\mathsf{LZ_{gr}}}{\mathsf{LZ_{opt}}} \ge \Omega(\log n)$.
\label{ExampleTheorem2}
\end{theorem}
\begin{proof}
Let us first discuss a high-level idea of our construction. Consider the following string:
$$
t\cdot b_1cb'_1\cdot b_2cb'_2\cdots b_kcb'_k\cdot c^{\Theta(n)}\cdot td\cdot b_1cb'_1d\cdot b_2cb'_2d\cdots b_kcb'_kd,
$$
where $t = a_1a_2\cdots a_{\sigma-2}$ is a string consisting of $\sigma{-}2$ distinct letters, the sets $\{b_i, b'_i\}$ run through all $k = {\sigma{-}2 \choose 2}$ two-element subsets of the set $\{a_1, a_2, \ldots, a_{\sigma-2}\}$, and $c$ and $d$ are two special letters with constant codes (say, $0$ and $1$) that do not occur in $t$. The greedy LZ77 parser parses the suffix $b_1cb'_1d\cdot b_2cb'_2d\cdots b_kcb'_kd$ into phrases $b_icb'_id$ encoded by references to the substrings $b_icb'_i$ of the prefix $t\cdot b_1cb'_1\cdot b_2cb'_2\cdots b_kcb'_k$. Each such reference takes $\Omega(\log n)$ bits and, therefore, the greedy encoding occupies $\Omega({\sigma{-}2 \choose 2}\log n) = \Omega(\sigma^2\log n)$ bits.

Obviously, any LZ77 encoding spends $\Theta(\log n)$ bits to encode the substring $c^{\Theta(n)}$. If we were able to encode the prefix and the suffix surrounding the substring $c^{\Theta(n)}$ in $O(\sigma^2)$ bits, then we would obtain $\frac{\mathsf{LZ_{gr}}}{\mathsf{LZ_{opt}}} \ge \Omega(\frac{\sigma^2\log n}{\sigma^2 + \log n}) = \Omega(\frac{\sigma^2\log n}{\max\{\sigma^2, \log n\}}) = \Omega(\min\{\log n, \sigma^2\})$, which is $\Omega(\log n)$ for $\sigma = \Omega(\sqrt{\log n})$.
Unfortunately, it seems that the best encoding that one can find for the suffix $b_1cb'_1d\cdot b_2cb'_2d\cdots b_kcb'_kd$ parses each substring $b_icb'_id$ into two phrases $b_ic$ and $b'_ic$, encoding each of them by a reference to a letter in $t = a_1a_2\cdots a_{\sigma-2}$, thus spending $\Theta({\sigma{-}2 \choose 2}\log {\sigma{-}2 \choose 2}) = \Theta(\sigma^2\log\sigma)$ bits for the whole suffix, which is larger than $\Theta(\sigma^2)$ by the factor $\log\sigma$. To address this issue, we construct a more sophisticated string equipped with additional ``infrastructure'' that helps to ``deliver'' cheaply letters from a ``dictionary'' substring (like $t$) to the places where these letters are used. Let us formalize this intuition.

Choose the minimal positive integer $x$ such that $2^{2^x} > \sqrt{\log n}$. The alphabet for our example will consist of two special letters $c$ and $d$ with codes $0$ and $1$, and of the set $A$ of $2^{2^x}$ letters with codes larger than $1$. Obviously, the alphabet size $\sigma = 2^{2^x} + 2$ is at most $\log n + 2$.

Let us assign to each subset $S$ of $A$ such that $|S| = 2^{2^i}$, for some $i \in [1..x]$, a Steiner system $S(2, 2^{2^{i-1}}, 2^{2^i})$ with the set of blocks denoted by $B_S$. Denote by $q$ a mapping that maps every such $S$ to a string $q(S) = a_{j_1}da_{j_2}d\cdots a_{j_{|S|}}d$, where $a_{j_1}, a_{j_2}, \ldots, a_{j_{|S|}}$ are all letters from $S$ in an arbitrarily chosen order. The basic building elements for our string are defined recursively as follows.
$$
\begin{array}{l}
r(S) = q(S)\prod_{B \in B_S} r(B) \quad\text{ if }|S| > 2,\\
r(S) = bcb'cbcb'dd \quad\text{ if }S = \{b, b'\}\text{ for distinct letters }b, b'.
\end{array}
$$
Analogously, we define:
$$
\begin{array}{l}
r'(S) = q(S)\prod_{B \in B_S} r(B) \quad\text{ if }|S| > 2,\\
r'(S) = bcb'cbcb'dc \quad\text{ if }S = \{b, b'\}\text{ for distinct letters }b, b'.
\end{array}
$$
To break ties on the lowest levels of recursion where $|S| = 2$, we assume that $b$ is the letter from $S$ with the smallest code.

Our string on which $\frac{\mathsf{LZ_{gr}}}{\mathsf{LZ_{opt}}} \ge \Omega(\log n)$ is $s = r'(A)c^\ell r(A)$, where $\ell$ is chosen so that $\ell = \Theta(n)$ (see the text below, where we discuss the lengths of $r(S)$ and $r'(S)$). Let us first show that the greedy LZ77 encoding of this string has size $\Omega(\sigma^2\log n)$ bits.

By the definition of Steiner systems, for any subset $S\subseteq A$ of size $2^{2^i}$, each pair $\{b, b'\}$ of distinct letters from $S$ is contained in exactly one block (of size $2^{2^{i-1}}$) from $B_S$. Then, it is straightforward that any given pair $\{b, b'\}$ of distinct letters from $A$ occurs exactly once as a parameter of $r$ on the lowest level of the recursion $r(A)$. An analogous claim holds for $r'(A)$. Hence, the string $bcb'cbcb'd$ (we assume that the code of $b$ is smaller than the code of $b'$) occurs in $s$ exactly twice: in the prefix $r'(A)$ and in the suffix $r(A)$. Further, it is easy to see that the string $bcb'$ occurs in $s$ only as a substring of $bcb'cbcb'd$. 
By a straightforward case analysis, one can show that this implies that the greedy LZ77 parsing of $s$ has a phrase $f$ containing the substring $bcb'dd$ of $r(A)$: $f$ either is a phrase starting at one of the first five positions of $bcb'cbcb'dd$ (greedily ``eating'' the remaining part) or is a phrase containing the prefix $bcb'cb$ of $bcb'cbcb'dd$ (the part $bcb'c$ can be copied only from $bcb'cbcb'dc$ in $r'(A)$ and, thus, again $f$ greedily ``eats'' the remaining part). The encoding of $f$ copies the part $bcb'd$ from the substring $bcb'd$ of $r'(A)$ by reference, thus spending $\Omega(\log\ell) = \Omega(\log n)$ bits. Since the two occurrences of $bcb'cbcb'd$ in $s$ are followed by distinct letters ($c$ in $r'(A)$ and $d$ in $r(A)$), the string $bcb'dd$ must be a suffix of $f$. Hence, there is a one-to-one correspondence between the pairs $\{b,b'\}$ of distinct letters from $A$ and the phrases containing the substrings $bcb'dd$. Therefore, the greedy LZ77 encoding of $s$ occupies $\Omega({|A| \choose 2}\log n) = \Omega(\sigma^2\log n)$ bits.

Now it remains to show that there is an LZ77 encoding of the string $s$ that occupies $O(\sigma^2 + \log n)$ bits. This will imply that $\frac{\mathsf{LZ_{gr}}}{\mathsf{LZ_{opt}}} \ge \frac{\Omega(\sigma^2\log n)}{O(\sigma^2 + \log n)} \ge \Omega(\min\{\log n, \sigma^2\})$, which is $\Omega(\log n)$ since, by construction, $\sigma > \sqrt{\log n}$.

We decompose the substring $c^\ell$ of $s = r'(A)c^\ell r(A)$ into two phrases $c$ and $c^{\ell-1}$, encoding these phrases in $O(\log n)$ bits. All other phrases in our parsing will have length either one or two. For simplicity of the exposition, we consider only encoding of the suffix $r(A)$; the encoding for $r'(A)$ is analogous and occupies asymptotically the same space.

By definition, $q(A)$ is a prefix of $r(A)$. The string $q(A)$ serves as a ``dictionary'' of letters similar to the string $t$ in the preliminary example. We encode each letter of $q(A)$ as a phrase of length one, thus spending $O(\sigma\log\sigma)$ bits. These are the only ``heavy'' phrases of length one in our encoding of $r(A)$: all other phrases of length one will be either $c$ or $d$, the letters with codes $0$ and $1$, which can be encoded in $O(1)$ bits. All phrases of length two will have the form either $ac$ or $ad$, where $a \in A$; thus, the ``heavy'' part of the encoding of such phrases of length two is an \mbox{$O(\log\delta)$-bit} encoding of the distance $\delta$ to an occurrence of $a$ preceding this phrase.

Let us consider a substring $r(S) = q(S)\prod_{B \in B_S} r(B)$ of $r(A)$, where $S \subseteq A$ is a set of size $2^{2^i}$ that occurs in the expansion of the recursion $r(A)$. Suppose that $i > 1$. Then, each substring $r(B)$, for $B \in B_S$, has a prefix $q(B) = a_1da_2d\cdots a_{|B|}d$, where $a_1, a_2, \ldots, a_{|B|}$ are members of $B$. We parse $q(B)$ into phrases $a_1d, a_2d, \ldots, a_{|B|}d$, encoding each phrase $a_id$ by a reference to the letter $a_i$ of the prefix $q(S)$ of $r(S)$. Suppose that $i = 1$. Then, each block $B \in B_S$ is just a pair $\{b,b'\}$ of distinct letters from $S$, and $r(B) = bcb'cbcb'dd$. We parse $r(B)$ into phrases $bc, b'c, bc, b'd, d$, encoding each phrase of length two by a reference to a letter from the prefix $q(S)$ of the string $r(S)$.

Denote by $E(i)$ the maximum size in bits of the encoding for the suffix $\prod_{B \in B_S} r(B)$ of some string $r(S)$, among all subsets $S \subseteq A$ such that $|S| = 2^{2^i}$. Then, $E(i)$ can be expressed by the following recursion (recall that $|B_S| = {2^{2^i} \choose 2} / {2^{2^{i-1}} \choose 2}$):
$$
\begin{array}{l}
E(i) \le \left({2^{2^i} \choose 2} / {2^{2^{i-1}} \choose 2}\right) (2^{2^{i-1}} \alpha\log L(i) + E(i-1)),\quad\text{ for }i > 1,\\
E(1) \le {4 \choose 2} (4\alpha\log L(1) + \alpha),
\end{array}
$$
where $L(i)$ denotes the length of the string $r(S)$ (obviously, $L$ depends only on the size $2^{2^i}$ of $S$) and $\alpha$ is a positive constant that depends on the chosen phrase encoder. Consider the prefix $q(B)$ of a substring $r(B)$ of $r(S)$, where $B \in B_S$ and $|B| > 2$. Each phrase $ad$ from the parsing of $q(B)$ is encoded in $O(\log\delta)$ bits, where $\delta$ is the distance to the letter $a$ from the prefix $q(S)$ of $r(S)$. Obviously, we have $\delta < L(i)$. Therefore, choosing an appropriate constant $\alpha > 0$, we can estimate the number of bits required to encode all $2^{2^{i-1}}$ phrases from the parsing of $q(B)$ as $2^{2^{i-1}} \alpha\log L(i)$; hence, the expression for $E(i)$ with $i \ne 1$. Analogously, the size in bits of the encoding for $bcb'cbcb'dd$ can be estimated as $4\alpha\log L(1) + \alpha$; hence, the expression for $E(1)$.

Thus, the whole encoding of the string $s$ requires $O(\log n + \sigma\log\sigma + E(x))$ bits. It remains to show that $E(x) \le O(\sigma^2)$. Before finding a closed form for $E(i)$, let us consider $L(i)$, which can be expressed by the following recursion:
$$
\begin{array}{l}
L(i) = 2\cdot 2^{2^i} + \left({2^{2^i} \choose 2} / {2^{2^{i-1}} \choose 2}\right) L(i-1),\quad\text{ for }i > 0,\\
L(0) = 9.
\end{array}
$$
Here, $L(0) = |bcb'cbcb'dd| = 9$. Let us find a closed form for $L(i)$. Note that $2^{2^z} / {2^{2^z} \choose 2} = \frac{2}{2^{2^z} - 1}$ for any integer $z \ge 0$. Expanding the recursion for $L(i)$, we obtain:
$$
\begin{array}{l}
L(i) = 2\cdot 2^{2^i} + \frac{{2^{2^i} \choose 2}}{{2^{2^{i-1}} \choose 2}} L(i-1)\\
= 2^{2^i+1} + \frac{{2^{2^i} \choose 2}}{{2^{2^{i-1}} \choose 2}} \left(2\cdot 2^{2^{i-1}} + \frac{{2^{2^{i-1}} \choose 2}}{{2^{2^{i-2}} \choose 2}} L(i - 2)\right)\\
 = 2^{2^i+1} + \frac{4\cdot {2^{2^i} \choose 2}}{2^{2^{i-1}} - 1} + \frac{{2^{2^i} \choose 2}}{{2^{2^{i-2}} \choose 2}} L(i-2)\\
 = 2^{2^i+1} + \left(\frac{4\cdot {2^{2^i} \choose 2}}{2^{2^{i-1}} - 1} + \frac{4\cdot {2^{2^i} \choose 2}}{2^{2^{i-2}} - 1} + \cdots + \frac{4\cdot {2^{2^i} \choose 2}}{2^{2^{1}} - 1}\right) + 9\cdot {2^{2^i} \choose 2}\\
 = 2^{2^i+1} + {2^{2^i} \choose 2}\left(\frac{4}{2^{2^{i-1}} - 1} + \frac{4}{2^{2^{i-2}} - 1} + \cdots + \frac{4}{2^{2^{1}} - 1} + 9\right).
\end{array}
$$
The term $9\cdot {2^{2^i} \choose 2}$ appears because of the last level of the recursion $L(i)$. Now it is easy to see that $L(i) \le \beta\cdot {2^{2^i} \choose 2}$ for a constant $\beta > 0$. In particular, we obtain $|r(A)| = |r'(A)| = L(x) \le \beta\cdot{2^{2^x} \choose 2} \le O(\sigma^2)$ (recall that $\sigma = 2^{2^x} + 2$). Since, as it was noted above, $\sigma \le \log n + 2$, we obtain $L(x) \le O(\log^2 n)$. Hence, for large enough $n$, we have $\ell = n - |r(A)| - |r'(A)| = n - 2L(x) = n - O(\log^2 n) \ge \tfrac{1}2 n$, i.e., $\ell = \Theta(n)$, as it was announced above. Let us similarly estimate $E(i)$. Denote $\gamma_i = \alpha\log L(i)$ for brevity.
$$
\begin{array}{l}
E(i) \le \frac{{2^{2^i} \choose 2}}{{2^{2^{i-1}} \choose 2}} (2^{2^{i-1}} \gamma_i + E(i-1))\\
 = \frac{2\cdot{2^{2^i} \choose 2}}{2^{2^{i-1}} - 1}\gamma_i + \frac{{2^{2^i} \choose 2}}{{2^{2^{i-1}} \choose 2}} E(i-1)\\
 = \frac{2\cdot{2^{2^i} \choose 2}}{2^{2^{i-1}} - 1}\gamma_i + \frac{{2^{2^i} \choose 2}}{{2^{2^{i-1}} \choose 2}}\left(\frac{{2^{2^{i-1}} \choose 2}}{{2^{2^{i-2}} \choose 2}} (2^{2^{i-2}} \gamma_{i-1} + E(i-2))\right)\\
 = \frac{2\cdot{2^{2^i} \choose 2}}{2^{2^{i-1}} - 1}\gamma_i + \frac{{2^{2^i} \choose 2}}{{2^{2^{i-2}} \choose 2}} 2^{2^{i-2}} \gamma_{i-1} + \frac{{2^{2^i} \choose 2}}{{2^{2^{i-2}} \choose 2}} E(i-2) \\
 = \frac{2\cdot{2^{2^i} \choose 2}}{2^{2^{i-1}} - 1}\gamma_i + \frac{2\cdot{2^{2^i} \choose 2}}{2^{2^{i-2}} - 1} \gamma_{i-1} + \frac{{2^{2^i} \choose 2}}{{2^{2^{i-2}} \choose 2}} E(i-2)\\
 = \frac{2\cdot{2^{2^i} \choose 2}}{2^{2^{i-1}} - 1}\gamma_i + \frac{2\cdot{2^{2^i} \choose 2}}{2^{2^{i-2}} - 1} \gamma_{i-1} + \cdots + \frac{2\cdot{2^{2^i} \choose 2}}{2^{2^{1}} - 1} \gamma_{2} + {2^{2^i} \choose 2}(4\gamma_1{+}\alpha)\\
 = 2\cdot{2^{2^i} \choose 2}\left(\frac{\gamma_i}{2^{2^{i-1}} - 1} + \frac{\gamma_{i-1}}{2^{2^{i-2}} - 1} + \cdots + \frac{\gamma_2}{2^{2^{1}} - 1} + 2\gamma_1 + \frac{\alpha}2\right).
\end{array}
$$
The term ${2^{2^i} \choose 2}(4\gamma_1 + \alpha)$ appear because of the last level of the recursion $E(i)$. Note that $\gamma_i = \alpha\log L(i) \le \alpha\log(\beta\cdot{2^{2^i} \choose 2}) = O(2^{i})$. It is well known that $\sum_{k=0}^\infty \frac{2^k}{2^{2^k} - 1} = O(1)$. Therefore, $E(i)$ can be estimated as $O({2^{2^i} \choose 2})$. Thus, we obtain $E(x) \le O({2^{2^x} \choose 2})$, which is $O(\sigma^2)$ since $\sigma = 2^{2^x} + 2$.
\end{proof}

\begin{remark}
For nonclassical LZ77 encodings, we can use exactly the same example as in the proof of Theorem~\ref{ExampleTheorem2}. In this case, the substrings $q(B) = a_1da_2d\cdots a_{|B|}d$ and $bcb'cbcb'dd$ of each string $r(S)$ are parsed into one-letter phrases: the phrases $c$ and $d$ are encoded in $O(1)$ bits using the codes of these letters, and the phrases $a_1, a_2, \ldots, a_{|B|}, b, b'$ are encoded using references to letters of the prefix $q(S)$ of $r(S)$. The analysis of the size of thus obtained encoding is analogous.
\end{remark}

\section{Concluding Remarks}\label{SectConclusion}

The upper and lower bounds $O(\min\{z, \frac{\log n}{\log\log_\sigma z}\})$ and $\Omega(\min\{z, \frac{\log n}{\log\log_\sigma z + \log\sigma}\})$, established in Theorems~\ref{MainTheorem} and~\ref{ExampleTheorem}, completely solve the problem for the case of constant alphabets and for some cases of arbitrary alphabets. But the general case of arbitrary alphabets with bounds expressed in terms of the parameters $n, z, \sigma$ remains open (see Table~\ref{tbl:results} in the introduction). Note that the examples constructed in the proof of Theorem~\ref{ExampleTheorem2} to show that $\frac{\mathsf{LZ_{gr}}}{\mathsf{LZ_{opt}}} \ge \Omega(\log n)$ are extremely compressible strings with $z = O(\log^2 n)$ and it is not clear whether the upper bound $\frac{\mathsf{LZ_{gr}}}{\mathsf{LZ_{opt}}} \le O(\log n)$ remains tight if we consider ``non-extremely compressible'' strings (but not necessarily on polylogarithmic alphabets).


It is interesting to consider other encoders for LZ77. Many practical compressors utilize a type of phrase encoders that is strikingly different from ours: such encoders use entropy compression as a component. DEFLATE and LZMA are important examples of compression schemes using such techniques. This is a major open problem to formalize these schemes and to conduct a similar theoretical analysis of the efficiency of the popular greedy approach.

\bibliography{lz77-enc}
\end{document}